\documentclass[12pt,oneside]{report}

\usepackage[letterpaper,left=1in,right=1in,top=1in,bottom=1in]{geometry}
\usepackage{setspace}
\usepackage{braket}
\usepackage{verbatim}

\usepackage{newtxtext,newtxmath}
\usepackage{microtype}
\usepackage{graphicx}
\usepackage{amsmath}

\usepackage{amssymb}

\usepackage{amsthm}
\usepackage{thmtools}
\usepackage{enumitem}
\usepackage{caption}
\usepackage{etoolbox}
\usepackage{indentfirst}

\usepackage{fancyhdr}
\usepackage{titlesec}

\titleformat{\chapter}[block]
  {\centering\bfseries}
  {\MakeUppercase{\chaptername\ \thechapter:}}
  {0.75em}
  {\MakeUppercase}
\titlespacing*{\chapter}{0pt}{0pt}{1.5em}

\titleformat{name=\chapter,numberless}[block]
  {\centering\bfseries\MakeUppercase}
  {}
  {0pt}
  {}
\titlespacing*{name=\chapter,numberless}{0pt}{0pt}{1.5em}

\usepackage{tocloft}
\usepackage[hidelinks]{hyperref}

\renewcommand{\cftchappresnum}{Chapter\ }
\renewcommand{\cftchapaftersnum}{:}
\newcommand{\ThesisTitleLineOne}{Quantum Dynamics: A Dilation-Based Approach}
\newcommand{\ThesisTitleLineTwo}{Finite-dimensional dilations of reduced quantum dynamics}
\newcommand{\AuthorName}{Caleb A. Mickelson, B.S.} 
\newcommand{\DegreeName}{MASTER OF SCIENCE}

\newcommand{\CommitteeChair}{Duy Nguyen Vu Hoang, Ph.D., Chair}
\newcommand{\CommitteeMemberA}{Mostafa Fazly, Ph.D.}
\newcommand{\CommitteeMemberB}{Jose Morales, Ph.D.}

\newcommand{\CollegeName}{College of Sciences}
\newcommand{\DepartmentName}{Department of Mathematics} 
\newcommand{\MonthYear}{May 2026} 

\newcommand{\UTSATitlePage}{%
  \begin{titlepage}
    \thispagestyle{empty}
    \begin{center}
      \vspace*{0in}

      {\bfseries
      \begin{spacing}{2}
      \MakeUppercase{\ThesisTitleLineOne}\par
      \MakeUppercase{\ThesisTitleLineTwo}\par
      \end{spacing}
      }

      \vspace{1in}
      by\par
      \vspace{1.0em}

      {\MakeUppercase{\AuthorName}\par}

      \vspace{2.0em}
      
      {\MakeUppercase{THESIS}\par}
      \begin{spacing}{1}
      \vspace{0.8em}
      Presented to the Graduate Faculty of\par
      The University of Texas at San Antonio\par
      in Partial Fulfillment\par
      of the Requirements\par
      for the Degree of\par
      \vspace{1.2em}
      {\MakeUppercase{\DegreeName}\par}

      \vspace{1.8em}

      \MakeUppercase{COMMITTEE MEMBERS:}\par
      \CommitteeChair\par
      \CommitteeMemberA\par
      \CommitteeMemberB\par

      \vfill

      \MakeUppercase{THE UNIVERSITY OF TEXAS AT SAN ANTONIO}\par
      \CollegeName\par
      \DepartmentName\par
      \MonthYear\par
      \end{spacing}
    \end{center}
  \end{titlepage}
  \clearpage
}

\newcommand{\UTSACopyrightPage}{%
\begin{spacing}{1}
  \thispagestyle{empty}
  \begin{center}
    \vspace*{\fill}
    Copyright \the\year\ Caleb Mickelson\\
    All rights reserved.
    \vspace*{\fill}
  \end{center}
  \end{spacing}
  \clearpage
}

\newcommand{\UTSAFrontChapterNoNumber}[1]{%
  \chapter*{#1}
  \thispagestyle{empty}
}

\newcommand{\UTSAFrontChapter}[2]{%
  \chapter*{#1}
  \addcontentsline{toc}{chapter}{#2}
  \thispagestyle{fancy}
}

\newcommand{\UTSAVitaTOCEntry}{%
  \addtocontents{toc}{\protect\noindent Vita\protect\par}
}

\newcommand{\UTSAListOfDefsThms}{%
  \cleardoublepage
  \phantomsection
  \addcontentsline{toc}{chapter}{List of Definitions and Theorems}
  \listoftheorems[ignoreall,show={definition,theorem,lemma,proposition,corollary}]
  \thispagestyle{fancy}
  \clearpage
}

\newcommand{\UTSAListOfFigures}{%
  \cleardoublepage
  \phantomsection
  \addcontentsline{toc}{chapter}{List of Figures}
  \listoffigures
  \thispagestyle{fancy}
  \clearpage
}

\newtheorem{theorem}{Theorem}[chapter]
\newtheorem{lemma}[theorem]{Lemma}
\newtheorem{proposition}[theorem]{Proposition}

\theoremstyle{definition}
\newtheorem{definition}[theorem]{Definition}

\theoremstyle{remark}

\newcommand{\stepparagraph}[1]{\par\smallskip\noindent\textbf{#1}\par\noindent}

\begin{document}

\UTSATitlePage

\pagenumbering{roman}
\setcounter{page}{2} 

\UTSACopyrightPage

\UTSAFrontChapterNoNumber{DEDICATION}
\begin{center}
\itshape
For Berkley
\end{center}
\clearpage

\UTSAFrontChapter{ACKNOWLEDGMENTS}{Acknowledgments}
I would like to express my sincere gratitude to my advisor, Dr.\ Duy Nguyen Vu Hoang, for his guidance, support, and patience throughout the development of this thesis. I am also grateful to my committee members, Dr.\ Mostafa Fazly and Dr.\ Jose Morales, for their time, thoughtful feedback, and encouragement. I am also grateful to Dr.\ Frederik vom Ende, whose work strongly influenced the direction of this thesis and whose email correspondence was helpful in shaping the project. I would further like to thank the faculty and staff of the Department of Mathematics at The University of Texas at San Antonio for their support during my graduate studies. Finally, I am especially grateful to my wife for her love, encouragement, and constant support throughout this process.

\vfill
\begin{center}
May 2026
\end{center}
\clearpage

\cleardoublepage
\phantomsection
\addcontentsline{toc}{chapter}{Abstract}
\thispagestyle{fancy}

\begin{center}
{\bfseries \MakeUppercase{\ThesisTitleLineOne}\par}
{\bfseries \MakeUppercase{\ThesisTitleLineTwo}\par}
\vspace{1em}
Caleb A. Mickelson, M.S.\par
The University of Texas at San Antonio, 2026\par
Supervising Professor: Duy Nguyen Vu Hoang, Ph.D.\par
\end{center}

\noindent
In the study of open quantum systems, one commonly describes the evolution of a system of interest through reduced dynamics, obtained by treating the environment indirectly rather than as a part of the full model. This thesis presents an expository account of an alternative, dilation-based viewpoint in the finite-dimensional setting, where a family of reduced dynamics is represented through unitary evolution on a larger system consisting of the original system together with an ancillary environment. After reviewing the reduced-dynamics perspective and the language of quantum channels, we formulate finite-dimensional quantum dynamics as channel-valued dynamical curves and use this framework to discuss Stinespring dilations of such curves. We then present exact dilation results for analytic dynamical curves, explain the singular behavior that can arise at $t=0$, and describe approximation results showing that Lipschitz-continuous dynamical curves admit approximate finite-dimensional Stinespring dilations. The thesis therefore provides a mathematically focused introduction to dilation-based modeling of quantum dynamics and argues that a change of perspective can lead to new ways of formulating problems in the theory of open quantum systems.

\clearpage

\tableofcontents
\clearpage

\UTSAListOfDefsThms
\UTSAListOfFigures

\cleardoublepage
\pagenumbering{arabic}

\chapter{Introduction}
Quantum systems are often described in two very different ways. On the one hand, the evolution of a closed quantum system is unitary, so in principle the dynamics are completely reversible and governed by the Schr\"odinger equation. On the other hand, the systems that appear in practice are rarely isolated. They interact with surrounding degrees of freedom, exchange information with an environment, and are therefore modeled as open systems. In that setting, the effective dynamics of the system of interest are typically no longer written as unitary evolution on the system alone, but instead as reduced dynamics, often expressed through a master equation.\\

The difference is not really a contradiction. It reflects two perspectives on the same problem. If one wants a fundamental description, then it is natural to model the joint system-environment evolution as unitary on a larger Hilbert space. If one wants a practical description of only the subsystem being observed, then it is natural to eliminate the environment and work with an effective non-unitary evolution on the system itself. The central idea behind this thesis is that these two viewpoints are systematically related by dilation theory, and that instead of beginning with a full system-environment model and tracing out the environment, one may ask when a given reduced dynamics can be realized by unitary dynamics on a larger space.\\

The main goal of this thesis is to present a clear finite-dimensional account of this dilation-based viewpoint for quantum dynamics. We work with quantum channels and channel-valued dynamical curves as a common language for describing both closed and open system evolution. Within this framework, the thesis explains how Stinespring dilations can be used to reconstruct enlarged unitary models from reduced dynamics, at least in a mathematically controlled sense. The emphasis is not on introducing a new theory, but on organizing a collection of recent results into a coherent narrative that is accessible to the reader with standard background in linear algebra, operator theory, and basic quantum mechanics.\\

Our discussion is restricted throughout to the finite-dimensional setting. This choice is deliberate. First, it allows the main constructions to be stated cleanly in terms of matrices, quantum channels, and finite ancillary systems. Second, it is the setting most compatible with explicit examples and with the broader motivation of modeling and simulation. In particular, one of the recurring themes of the thesis is that dilation-based descriptions can replace a potentially complicated or inaccessible environment by a finite-dimensional ancillary system that still reproduces the reduced dynamics of interest.\\

After reviewing the standard reduced-dynamics approach to open quantum systems, we introduce quantum dynamical curves and show how they provide a unified language for closed-system unitary evolution and open-system dynamics such as quantum dynamical semigroups. We then turn to the main dilation results discussed in the thesis: analytic dynamical curves admit Stinespring dilations with corresponding regularity for positive times; this construction may exhibit a singularity at $t=0$; and more recent approximation results show that Lipschitz-continuous dynamical curves can still be approximated by finite-dimensional Stinespring dilations. Together, these results show that the dilation-based perspective is broad enough to capture a substantial class of physically relevant reduced dynamics while remaining mathematically explicit in finite dimensions.\\

The point of the thesis, then, is not to argue that reduced dynamics should be abandoned. It is to show that they can often be lifted to a larger unitary picture, and that this larger picture can be studied in a precise and meaningful way. From that perspective, dilation theory gives a bridge between the effective description used for open systems and the unitary framework that underlies quantum mechanics more broadly.\\

The thesis is organized as follows. Chapter 2 reviews the necessary background on open quantum systems and finite-dimensional Lindblad dynamics. Chapter 3 sets up the channel-based perspective. Chapter 4 formulates the dilation problem, then discusses exact analytic dilations and the singular behavior that can occur at the initial time, and concludes with a treatment of approximate dilations for regular channel evolutions, especially in the finite-dimensional setting considered throughout this work.

\chapter{Open Quantum Systems}
Before turning to dilation-based descriptions of quantum dynamics, it is useful to recall the standard reduced-dynamics perspective used for open quantum systems. In this approach, one focuses on a subsystem of interest and treats the surrounding degrees of freedom indirectly, so that the effective evolution of the subsystem is no longer described solely by unitary dynamics on its own state space. This viewpoint provides the usual mathematical framework for modeling dissipation, decoherence, and irreversible behavior in quantum systems.\\

Our goal here is modest. In $\textsection2.1$, we explain the reduced-dynamics viewpoint and its role in the study of open systems. In $\textsection2.2$, we introduce the GKLS/Lindblad equation as the standard description of Markovian open-system dynamics. The chapter is meant only to establish background and notation for what follows, not to give a complete treatment of open quantum systems. The main shift in perspective comes in the next chapter, where quantum dynamics will be reformulated in terms of channel-valued curves.

\section{The Reduced-Dynamics Perspective}
In the idealized description of a closed quantum system, the state of the system evolves unitarily in time. For many systems of physical interest, however, this is not the most useful description. Real quantum systems are typically not isolated. They interact with surrounding degrees of freedom, exchange information with an external medium, and are influenced by measurement devices, control mechanisms, or ambient noise. In such situations, it is often neither realistic nor desirable to model every degree of freedom present in the full physical setup. Instead, one singles out a subsystem of interest and seeks an effective description of its time evolution alone.\\

This leads to the reduced-dynamics perspective. One begins with a larger quantum system consisting of the subsystem of interest together with an environment. The combined system is treated as a closed system, so that its evolution is unitary on the larger Hilbert space. The subsystem dynamics are then obtained by discarding the environmental degrees of freedom, that is, by passing from the joint state of the system and environment to the reduced state of the subsystem. In this way, the subsystem generally evolves non-unitarily, even though the total system evolves unitarily. The apparent loss of reversibility is therefore not a violation of quantum mechanics, but a consequence of describing only part of a larger closed system.\\

This point of view is standard in the theory of open quantum systems. It reflects a practical compromise between physical completeness and mathematical tractability. A full microscopic model of system and environment may be inaccessible, unnecessarily complicated, or simply irrelevant to the question being studied. In many applications, one is interested primarily in how the subsystem behaves, for instance whether coherence is lost, whether energy is exchanged with the surroundings, or whether the system relaxes toward some stationary state. The reduced state contains exactly the information needed for such questions, while suppressing details of the environment that may be unknown or experimentally uncontrollable.\\

At the same time, the reduced-dynamics approach comes with an important conceptual shift. Once the environment has been removed from the description, the subsystem can no longer be expected to follow a unitary equation of motion on its own state space. Its dynamics may instead display dissipation, decoherence, or irreversibility. Thus, the distinction between closed and open quantum systems is not a distinction between two incompatible theories, but between two levels of description. The first is a larger closed description for the total system, while the second is a reduced effective description for the subsystem. The latter is especially useful in applications, but it should be understood as arising from the former.\\

This perspective also explains why master equations play such a central role in the study of open systems. Rather than tracking the full unitary dynamics of system and environment , one seeks an evolution equation directly for the reduced state of the subsystem. In favorable regimes, this produces a tractable effective model for the open-system dynamics. The best-known example is the GKLS/Lindblad equation which describes Markovian reduced dynamics under suitable assumptions on the system-environment interaction. We will discuss this model in the next section, not as a complete account of open-system dynamics, but as the standard reduced description against which the later dilation-based viewpoint may be compared.\\

For the purposes of this thesis, the reduced-dynamics perspective serves mainly as a starting point. The central question pursued later is whether one can reverse this viewpoint in a controlled way: given a family of reduced dynamics, can one construct a larger finite-dimensional system in which those dynamics arise from unitary evolution?\\
Chapter 3 introduces the language needed to formulate this question precisely in terms of quantum channels and dynamical curves, and Chapter 4 discusses exact and approximate finite-dimensional dilation results. In that sense, the present section provides the physical and mathematical motivation for the rest of the thesis by emphasizing that open-system dynamics are often described only after the environment has been suppressed, whereas dilation theory asks how such reduced descriptions can be lifted back to an enlarged unitary model.

\section{The GKLS/Lindblad Description}
The reduced-dynamics perspective discussed in the previous section leads naturally to the problem of finding effective equations of motion for an open quantum system. In general, such equations can be quite complicated, since the environment may retain memory of its interaction with the system and later feed that information back into the system dynamics. A particularly important simplification occurs in the Markovian regime, where one assumes that the reduced evolution is memoryless in an appropriate sense. In this setting, the dynamics of the system are described by a quantum dynamical semigroup, and the corresponding generator takes a very specific form.\\

In finite dimensions, the standard description of Markovian open-system evolution is given by the Gorini-Kossakowski-Lindblad-Sudarshan equation, usually called the GKLS or Lindblad equation. It describes the evolution of a density operator $\rho$ on the system Hilbert space through a first-order differential equation of the form
\[
\frac{d}{dt}\rho=\mathcal{L}(\rho),
\]
where $\mathcal{L}$ is the generator of the reduced dynamics. The content of the GKLS theorem \cite{gorini_kossakowski_sudarshan_1976, lindblad_1976} is that, under the assumptions of complete positivity, trace preservation, and semigroup structure, the generator $\mathcal{L}$ must have the form
\[
\mathcal{L}(\rho)=-i[H,\rho]+\sum_j\left(L_j\rho L_j^\ast-\frac{1}{2}\left\{L_j^\ast L_j,\rho\right\}\right).
\]
Here $H$ is a self-adjoint operator on the system Hilbert space, and the operators $L_j$ describe the dissipative part of the evolution.\\

This decomposition has a clear interpretation. The commutator term $-i[H,\rho]$ is the Hamiltonian part of the dynamics and coincides with the usual unitary evolution of a closed quantum system. The remaining terms encode the influence of the environment on the subsystem and are responsible for genuinely open-system effects such as dissipation and decoherence. In particular, the operators $L_j$, often called Lindblad or jump operators, represent the channels through which the system exchanges information or energy with its surroundings. The anticommutator term ensures that the evolution remains trace preserving, while the full structure of the generator guarantees complete positivity of the resulting dynamics.\\

The importance of the GKLS equation lies in the fact that it provides a mathematically consistent and physically meaningful model for irreversible quantum evolution. It gives the canonical form of a Markovian master equation and appears throughout the theory of open quantum systems, quantum optics, and quantum control. At the same time, its scope should be understood correctly. The Lindblad form does not describe every possible open-system evolution. Instead, it describes those evolutions that are well approximated by a time-homogeneous, completely positive, trace-preserving semigroup. Thus it is best viewed as a particularly tractable and widely used class of reduced dynamics rather than as the general theory of open systems.\\

For the purposes of this thesis, the role of the GKLS/Lindblad equation is primarily contextual. It provides the standard reduced description of Markovian open-system dynamics and therefore serves as an important example of the kinds of evolutions one would like to study more systematically. In the next chapter, we move away from the master-equation viewpoint and introduce a more general finite-dimensional framework in which quantum dynamics are described as curves of quantum channels. This shift in language will make it possible to formulate the later dilation questions in a precise way and to compare reduced dynamics with enlarged unitary descriptions.

\chapter{Quantum Dynamical Curves}
In the previous chapter, quantum dynamics were discussed primarily from the perspective of open quantum systems and the Lindblad master equation. While that framework is well suited to the discussion of Markovian reduced dynamics, the present thesis requires a more general language in which quantum evolutions can be treated as objects in their own right. In finite dimensions, this is naturally accomplished by viewing the state evolution at each time as a quantum channel and the full time-dependent evolution as a curve of such channels.\\

The purpose of this chapter is to introduce that channel-based framework and to establish the terminology used throughout the remainder of the thesis. We begin by reviewing the definition of a quantum channel and the principal representations that will be relevant later, including Kraus, Choi, and Stinespring forms. We introduce the notation of a dynamical curve, which provides a unified way to describe time-dependent quantum processes in both closed and open settings. Within this framework, several important classes of dynamics appear naturally, including unitary evolutions, quantum dynamical semigroups, and Stinespring-type evolutions on enlarged Hilbert spaces.\\

This reformulation is not merely notational. It provides the setting in which the central problem of the thesis can be stated precisely: given a prescribed curve of quantum channels, under what conditions can it be realized as the reduced dynamics of a continuous unitary evolution on a larger finite-dimensional system?\\ 
Chapter 4 addresses this question in two directions. First, it studies exact analytic dilations and the singular behavior that may arise at the initial time. It then considers approximate dilation results for sufficiently regular channel evolutions. The present chapter therefore serves as the conceptual and mathematical bridge between the background material of Chapter 2 and the dilation theory developed in the remainder of the thesis.

\section{Quantum Channels and Their Representations}
We begin with the basic object used throughout the remainder of the thesis.

\begin{definition}[Quantum channel]
    Let $\mathcal{H}_1$ and $\mathcal{H}_2$ be finite-dimensional Hilbert spaces. A \emph{quantum channel} is a linear map
    \[
    \Phi:\mathcal{B}(\mathcal{H}_1)\to\mathcal{B}(\mathcal{H}_2)
    \]
    that is completely positive and trace preserving.
\end{definition}

In the finite-dimensional setting, quantum channels provide the natural mathematical description of physically admissible state transformations. If $\rho\in\mathcal{B}(\mathcal{H}_1)$ is a density operator, then $\Phi(\rho)\in\mathcal{B}(\mathcal{H}_2)$ is again a density operator. Complete positivity ensures that the map remains positive even when the system under consideration is viewed as part of a larger composite system, while trace preservation expresses conservation of total probability.\\

Since this thesis is concerned primarily with finite-dimensional dynamics on a fixed system space $\mathbb{C}^n$, we write
\[
\mathrm{CPTP}(n)
\]
for the collection of all completely positive, trace-preserving maps
\[
\Phi:\mathcal{B}(\mathbb{C}^n)\to\mathcal{B}(\mathbb{C}^n).
\]
This notation will be used throughout the remainder of the thesis.\\

Quantum channels include both closed- and open-system evolutions. In the closed case, the evolution of a state is given by conjugation with a unitary operator $U$, so that
\[
\Phi(\rho)=U\rho U^\dagger.
\]

In the open case, one generally obtains a non-unitary channel after discarding environmental degrees of freedom. Thus the channel formalism provides a common language in which both types of quantum evolution can be treated on equal footing.\\

For later purposes, it is important to describe a channel in several equivalent ways. Different representations emphasize different aspects of the same map: the Choi representation is especially convenient for questions of positivity and continuity, the Kraus representation gives an operator-sum description, and the Stinespring representation realizes the channel as the reduced action of a unitary evolution on a larger Hilbert space. Since the dilation problem studied later in this thesis moves repeatedly between these viewpoints, we review each of them here.\\

\subsection{The Choi Representation}
We first recall the Choi matrix of a channel.

\begin{definition}[Choi matrix]
    Let $\mathcal{H}=\mathbb{C}^n$ with orthonormal basis $\{e_i\}_{i=1}^n$, and let $\Phi\in\mathrm{CPTP}(n)$. The \emph{Choi matrix} of $\Phi$ is the operator
    \[
    J(\Phi)=\sum_{i,j=1}^n\ket{e_i}\bra{e_j}\otimes\Phi(\ket{e_i}\bra{e_j})\in\mathcal{B}(\mathbb{C}^n\otimes\mathbb{C}^n)
    \]
\end{definition}

Equivalently, $J(\Phi)$ may be viewed as the block matrix obtained by applying $\Phi$ to the matrix units $\ket{e_i}\bra{e_j}$. In finite dimensions, this construction identifies the channel with a positive operator on the tensor product space $\mathbb{C}^n\otimes\mathbb{C}^n$. The Choi matrix is particularly useful because complete positivity of $\Phi$ is equivalent to positivity of $J(\Phi)$. Trace preservation can likewise be expressed as a linear constraint on the partial trace of $J(\Phi)$.\\

For the present thesis, the main importance of the Choi representation is that it converts a channel into an ordinary matrix object. This will be useful later when regularity properties of channel-valued curves are studied through the corresponding regularity of their Choi matrices.

\subsection{The Kraus Representation}
A second standard description of a channel is the operator-sum representation.

\begin{definition}[Kraus representation]
    A quantum channel $\Phi\in\mathrm{CPTP}(n)$ is said to have a \emph{Kraus representation} if there exist operators $K_1,\dots,K_R\in\mathcal{B}(\mathbb{C}^n)$ such that
    \[
    \Phi(\rho)=\sum_{i=1}^RK_i\rho K_i^\dagger\hspace{6mm} \text{for all }\rho\in\mathcal{B}(\mathbb{C}^n),
    \]
    and
    \[
    \sum_{i=1}^RK_i^\dagger K_i=I.
    \]
    The operators $K_i$ are called \emph{Kraus operators} for $\Phi$.
\end{definition}

In finite dimensions, every quantum channel admits such a representation. Moreover, one may choose $R\leq n^2$. The Kraus representation is often the most concrete description of a channel, since it writes the action of $\Phi$ directly in terms of finitely many operators on the system space.\\

The Kraus description is not unique because distinct families of Kraus operators may determine the same channel. Nevertheless, the minimal number of Kraus operators needed to represent $\Phi$, often called the \emph{Kraus rank}, is an intrinsic quantity and agrees with the rank of the Choi matrix $J(\Phi)$.

\subsection{Passage from Choi to Kraus form}
Because the Choi matrix $J(\Phi)$ is positive, it admits a spectral decomposition
\[
J(\Phi)=\sum_{i=1}^R\lambda_i\ket{v_i}\bra{v_i}, \hspace{6mm}\lambda_i>0.
\]
If one defines
\[
\ket{\psi_i}=\sqrt{\lambda_i}\ket{v_i},
\]
then
\[
J(\Phi)=\sum_{i=1}^R\ket{\psi_i}\bra{\psi_i}.
\]
By reshaping each vector $\ket{\psi_i}$ into a matrix $K_i$ through the inverse of the vectorization map, one obtains Kraus operators satisfying
\[
\Phi(\rho)=\sum_{i=1}^RK_i\rho K_i^\dagger.
\]

Conversely, if a channel is given in Kraus form, then its Choi matrix may be reconstructed as
\[
J(\Phi)=\sum_{i=1}^R \mathrm{vec}(K_i)\mathrm{vec}(K_i)^\dagger.
\]
Thus the Choi and Kraus descriptions are equivalent and can be converted into one another constructively.

\subsection{The Stinespring representation}
The representation most relevant for the later chapters is the Stinespring form, since it expresses a channel as the reduced dynamics arising from a larger unitary evolution.

\begin{theorem}[Finite-dimensional Stinespring representation]
    Let $\Phi\in\mathrm{CPTP}(n)$. Then there exist a finite-dimensional ancillary Hilbert space $E$, a unit vector $\omega\in E$, and a unitary operator
    \[
    U:\mathbb{C}^n\otimes E\to\mathbb{C}^n\otimes E
    \]
    such that 
    \[
    \Phi(\rho)=\mathrm{Tr}_E\left(U(\rho\otimes\ket{\omega}\bra{\omega})U^\dagger\right)\hspace{6mm} \text{for all }\rho\in\mathcal{B}(\mathbb{C}^n).
    \]
\end{theorem}
This theorem shows that every quantum channel may be realized by adjoining an ancillary system, evolving the joint system unitarily, and then tracing out the ancillary degreees of freedom. In this sense, the Stinespring representation connects non-unitary reduced dynamics with ordinary unitary evolution on a larger Hilbert space. We will often refer to this as the \emph{static Stinespring theorem} in order to distinguish it from the time-dependent dilation questions considered in Section 3.2.\\

It is often convenient to separate this construction into two stages. First, one introduces an isometry
\[
V:\mathbb{C}^n\to\mathbb{C}^n\otimes E
\]
such that
\[
\Phi(\rho)=\mathrm{Tr}_E(V\rho V^\dagger).
\]
One may then extend the isometry $V$ to a unitary operator on a larger space. In finite dimensions this extension is always possible after choosing the ancillary space large enough.\\

Starting from Kraus operators $K_1,\dots,K_R$, one standard choice is to let $E\cong\mathbb{C}^R$ with orthonormal basis $\{f_i\}_{i=1}^R$, and define
\[
Vx=\sum_{i=1}^RK_ix\otimes f_i,\hspace{6mm}x\in\mathbb{C}^n.
\]
Then $V^\dagger V=I$ by the Kraus completeness relation, so $V$ is an isometry, and the corresponding partial trace formula reproduces the original channel. This construction makes explicit how the Kraus and Stinespring viewpoints are related.\\

The ancillary space appearing in this representation should not automatically be interpreted as a physically distinguished environment. At this stage it is best regarded as a mathematical device that realizes the channel as a reduced unitary evolution. The question of whether such representations can be chosen coherently along an entire time-dependent family of channels is more delicate, and it is precisely this question that leads to the dilation problem studied in the following chapter.\\

In summary, a quantum channel may be described equivalently by its action as a completely positive trace-preserving map, by its Choi matrix, by a Kraus operator-sum representation, or by a Stinespring dilation on a larger Hilbert space. These descriptions emphasize different features of the same object, and each will play a role in what follows. We now pass from individual channels to time-dependent families of channels, which will serve as the basic language for finite-dimensional quantum dynamics in the remainder of the thesis.

\section{Dynamical Curves}
The preceding section described quantum channels as the basic objects representing physically admissible state transformations in finite dimensions. In order to study quantum dynamics, however, one must pass from a single channel to a time-dependent family of channels. This leads naturally to the notion of a dynamical curve.\\

Throughout this section, let $\mathcal{I}\subseteq\mathbb{R}$ denote a time interval, typically of the form $\mathcal{I}=[0,T]$ for some $T>0$, and let $\mathcal{H}_S\cong\mathbb{C}^n$ be the system Hilbert space.

\begin{definition}[Dynamical curve]
    A \emph{dynamical curve} on $\mathcal{H}_S$ is a map
    \[
    \Phi:\mathcal{I}\to\mathrm{CPTP}(n),\hspace{6mm}t\mapsto\Phi_t,
    \]
    that assigns to each time $t\in\mathcal{I}$ a quantum channel
    \[
    \Phi_t:\mathcal{B(\mathcal{H}_S})\to\mathcal{B(\mathcal{H}_S}).
    \]
\end{definition}

Thus, for each fixed time $t$, the map $\Phi_t$ describes the state transformation from the initial time to time $t$, while the full family $(\Phi_t)_{t\in\mathcal{I}}$ describes the evolution as a whole. In this way, a dynamical curve provides a channel-valued description of quantum dynamics.\\

This point of view is sufficiently general to include both closed- and open-system evolutions. In the closed case, one obtains channels induced by unitary conjugation. In the open case, one may obtain more general completely positive trace-preserving maps, for example from reduced system-environment dynamics or from solutions of a Markovian master equation. The usefulness of the notion lies precisely in the fact that these apparently different situations can be treated in a common mathematical language.\\

When needed, one may impose additional regularity conditions on the map $t\mapsto \Phi_t$, such as continuity, differentiability, analyticity, or Lipschitz continuity with respect to a chosen operator norm. Chapter 4 will make such assumptions when studying exact and approximate dilation problems. At the present stage, however, the basic definition above is sufficient.

\subsection{Unitary curves}
The simplest example of a dynamical curve arises from unitary evolution on the system Hilbert space itself.

\begin{definition}[Unitary curve]
    A dynamical curve $(\Phi_t)_{t\in\mathcal{I}}$ is called a \emph{unitary curve} if there exists a family of unitary operators
    \[
    U_t:\mathcal{H}_S\to\mathcal{H}_S,\hspace{6mm}t\in\mathcal{I},
    \]
    such that
    \[
    \Phi_t(\rho)=U_t\rho U_t^\dagger \hspace{6mm} \text{for all }\rho\in\mathcal{B}(\mathcal{H}_S),\,t\in\mathcal{I}.
    \]
\end{definition}
Unitary curves represent the dynamics of closed quantum systems. If $U_t$ is generated by a time-dependent Hamiltonian $H(t)$, then the corresponding channel family describes the familiar Schr\"odinger evolution at the level of density operators. In particular, unitary curves form a distinguished subclass of dynamical curves in which no information is lost to an environment.\\

From the present perspective, the importance of unitary curves is twofold. First, they provide the canonical example of reversible quantum dynamics. Second, they serve as the target objects in the dilation problem, in which one seeks to realize more general channel-valued evolutions as reduced parts of unitary dynamics on a larger Hilbert space.

\subsection{Quantum dynamical semigroups}
A second important class of channel-valued evolutions is given by quantum dynamical semigroups, which arise naturally in the study of time-homogeneous Markovian open-system dynamics.

\begin{definition}[Quantum dynamical semigroup]
    A dynamical curve $(\Phi_t)_{t\geq0}$ is called a \emph{quantum dynamical semigroup} if
    \[
    \Phi_0=\mathrm{id},\hspace{6mm}\Phi_{t+s}=\Phi_t\circ\Phi_s\hspace{4mm}\text{for all }s,t\geq0,
    \]
    and the map $t\mapsto\Phi_t$ is continuous in $t$.
\end{definition}

The semigroup property expresses time-homogeneous evolution. Specifically, the transformation from time $0$ to time $t+s$ is obtained by first evolving for time $s$ and then for time $t$, independently of the initial starting point. In finite dimensions, such semigroups are precisely the channel-valued evolutions generated by Lindblad-type operators, as discussed in Chapter 2. Thus quantum dynamical semigroups provide the channel formulation of standard Markovian reduced dynamics.\\

It is worth emphasizing that not every dynamical curve is a semigroup. A general time-dependent evolution may fail to satisfy the composition law above, either because the underlying generator depends explicitly on time or because the dynamics are not Markovian. The notion of a dynamical curve is therefore broader than that of a quantum dynamical semigroup, and this greater flexibility is essential for the later dilation results.

\subsection{Stinespring curves}
The class of dynamical curves most directly connected with the main problem of this thesis is obtained by combining the static Stinespring representation from Section 3.1 with time dependence.\\

Informally, one would like to realize a family of channels $(\Phi_t)_{t\in\mathcal{I}}$ by choosing, for each time $t$, a unitary operator on a larger Hilbert space whose reduced action reproduces $\Phi_t$. This leads to the following notion.

\begin{definition}[Stinespring curve]
    Let $\mathcal{H}_S\cong\mathbb{C}^n$ be the system Hilbert space, let $E$ be a finite-dimensional ancillary Hilbert space, and let $\omega\in E$ be a fixed unit vector. A dynamical curve $(\Phi_t)_{t\in\mathcal{I}}$ is said to admit a \emph{Stinespring curve representation} if there exists a family of unitary operators
    \[
    U_t:\mathcal{H}_S\otimes E\to\mathcal{H}_S\otimes E,\hspace{4mm}t\in\mathcal{I},
    \]
    such that
    \[
    \Phi_t(\rho)=\mathrm{Tr}_E\left(U_t(\rho\otimes\ket{\omega}\bra{\omega})U_t^\dagger\right)\hspace{6mm}\text{for all }\rho\in\mathcal{B}(\mathcal{H}_S),\,t\in\mathcal{I}.
    \]
    A family $(U_t)_{t\in\mathcal{I}}$ satisfying this condition will be called a \emph{Stinespring curve} for $(\Phi_t)_{t\in\mathcal{I}}$.
\end{definition}
This definition should be compared carefully with the static Stinespring theorem for a single channel. For an individual channel $\Phi_t$, the theorem guarantees the existence of some dilation of this form. The nontrivial issue in the time-dependent setting is whether the pointwise dilations may be chosen coherently in time, and, if so, with what regularity.\\

This distinction is fundamental for the present thesis. The later chapters do not merely ask whether each channel $\Phi_t$ is dilatable individually, that follows already from the static theorem. Rather, they ask whether the entire curve $t\mapsto\Phi_t$ can be realized by a correspondingly regular curve $t\mapsto U_t$ on a larger finite-dimensional Hilbert space. It is precisely this stronger question that leads to the exact and approximate dilation results studied later.

\subsection{Remarks on regularity and scope}
At this point it is useful to note that the concept of a dynamical curve is intentionally general. No semigroup property is assumed, no differential equation is imposed, and no preferred physical origin is built into the definition. This flexibility is essential because this thesis is concerned not only with standard Markovian evolutions but with general channel-valued curves satisfying suitable regularity assumptions.\\

The finite-dimensional restriction is also important. Since in this setting, all relevant operator spaces are finite-dimensional, so questions of continuity, differentiability, analyticity, and approximation can be formulated without the additional technical complications that arise in infinite-dimensional systems. Since the main results discussed later concern finite-dimensional exact analytic dilations and finite-dimensional approximation of channel evolutions, this framework is the natural one for the thesis.\\

In summary, a dynamical curve is a time-dependent family of quantum channels, and several important subclasses arise naturally within this framework. Unitary curves describe closed-system evolution, quantum dynamical semigroups describe time-homogeneous Markovian open-system dynamics, and Stinespring curves describe unitary realizations of channel-valued evolutions on enlarged Hilbert spaces. This language provides the setting in which the central dilation problem can now be formulated more precisely.

\chapter{Dilation-Based Quantum Dynamics}
In the previous chapter, quantum dynamics were reformulated as curves of quantum channels, and the central question of the thesis was expressed in this language: given a prescribed dynamical curve, under what conditions can it be realized as the reduced dynamics of a unitary evolution on a larger finite-dimensional Hilbert space?\\
The purpose of the present chapter is to study this question from the perspective of Stinespring dilations and to present the principal exact and approximate results that motivate the dilation-based viewpoint adopted throughout the thesis.\\

The key issue is not whether each individual channel $\Phi_t$ admits a Stinespring representation. In finite-dimensions, that follows already from the static Stinespring theorem for a single completely positive trace-preserving map. The deeper problem is whether these pointwise dilations can be chosen coherently in time, so that a full dynamical curve $t\mapsto\Phi_t$ is realized by a correspondingly regular curve of unitaries on an enlarged system. As discussed in Chapter 3, this passage from individual representations to a genuine Stinespring curve is not automatic, and the relevant questions are therefore questions of regularity as much as existence.\\

This chapter considers that problem in two complementary directions. First, we examine exact dilation results for analytic dynamical curves. In this setting, one can construct finite-dimensional Stinespring dilations whose unitary evolutions inherit analyticity for positive times. However, this exact construction is accompanied by an important limitation. Specifically, even when the reduced dynamics are well behaved, the corresponding dilation may develop singular behavior at the initial time $t=0$. This phenomenon shows that exact dilation of a channel-valued evolution can be substantially more delicate than the existence of a static Stinespring representation for each fixed time.\\

We then turn to approximation results, which weaken the regularity assumptions on the prescribed dynamics and broaden the scope of the dilation picture. In particular, recent work shows that Lipschitz-continuous dynamical curves can be approximated arbitrarily well by finite-dimensional Stinespring curves with controlled regularity. From the perspective of this thesis, these approximation results are significant because they show that the dilation-based framework is not limited to the analytic setting and remains meaningful for a much larger class of finite-dimensional reduced dynamics.\\

Taken together, the results of this chapter clarify both the power and the limitations of the finite-dimensional dilation approach. Exact realizations are available under strong regularity hypotheses, but singular behavior at the initial time shows that such realizations are more complicated. Approximation results, on the other hand, demonstrate that one can still recover enlarged unitary models in a mathematically controlled sense even when exact regular dilations are unavailable. This is the sense in which dilation theory provides a bridge between reduced descriptions of open quantum systems and unitary dynamics on larger Hilbert spaces.

\section{Stinespring Dilations of Dynamical Curves}
The finite-dimensional Stinespring theorem shows that every quantum channel can be realized as the reduced action of a unitary operator on a larger Hilbert space. For a single channel, this provides a structural representation of the map in terms of an enlarged closed-system evolution together with a fixed ancillary state. From the perspective of quantum dynamics, however, the relevant object is not an isolated channel but a time-dependent family of channels. The question is therefore whether a prescribed dynamical curve can be realized by a correspondingly regular time-dependent family of unitary dilations on a larger finite-dimensional system.\\

This is the basic dilation problem studied in the present chapter. If $(\Phi_t)_{t\in\mathcal{I}}$ is a dynamical curve on the system Hilbert space $\mathcal{H}_S\cong\mathbb{C}^n$, then for each fixed time $t$, the static Stinespring theorem guarantees the existence of some ancillary space, some ancilla state, and some unitary operator whose reduced action reproduces $\Phi_t$. What is not automatic is whether these pointwise choices may be made coherently as $t$ varies, so that the resulting family of unitary operators has useful regularity properties such as continuity, differentiability, or analyticity. It is this more refined question, rather than mere pointwise dilatability, that underlies the exact and approximate results discussed below.\\

In order to state the problem clearly, we first isolate the notion of a Stinespring dilation for a dynamical curve in the finite-dimensional setting used throughout this thesis.

\begin{definition}[Stinespring dilation of a dynamical curve]
    Let $\mathcal{I}\subseteq\mathbb{R}$ be a time interval, $\mathcal{H}_S\cong\mathbb{C}^n$ be the system Hilbert space, and
    \[
    \Phi:\mathcal{I}\to\mathrm{CPTP}(n),\hspace{4mm}t\mapsto\Phi_t,
    \]
    be a dynamical curve. A \emph{Stinespring dilation} of $(\Phi_t)_{t\in\mathcal{I}}$ consists of
    \begin{itemize}
        \item a finite-dimensional ancillary Hilbert space $E$,
        \item a fixed unit vector $\omega\in E$, and
        \item a family of unitary operators
        \[
        U_t:\mathcal{H}_S\otimes E\to\mathcal{H}_S\otimes E,\hspace{4mm}t\in\mathcal{I},
        \]
    \end{itemize}
    such that
    \[
    \Phi_t(\rho)=\mathrm{Tr}_E\left(U_t(\rho\otimes\ket{\omega}\bra{\omega})U_t^\dagger\right)
    \]
    for all $\rho\in\mathcal{B}(\mathcal{H}_S)$ and all $t\in\mathcal{I}$.\\
    
    A family $(U_t)_{t\in\mathcal{I}}$ satisfying this condition will be called a dilating unitary curve for $(\Phi_t)_{t\in\mathcal{I}}$.
\end{definition}
This definition is formally similar to the definition of a Stinespring curve introduced in Chapter 3, but its role here is more specific. The point is not simply to name a class of dynamical curves on enlarged Hilbert spaces, but to formulate the inverse problem central to this thesis, namely the problem of whether a channel-valued evolution on the system can be realized by a unitary curve on a larger space whose reduced dynamics reproduce it. In this way, the dilation problem asks when a reduced description may be lifted back to a closed-system model in finite dimensions.\\

Two features of this problem should be emphasized. First, the ancillary space and ancilla state are required to remain fixed throughout the evolution. This is what allows the family $(U_t)_{t\in\mathcal{I}}$ to be interpreted as a single enlarged dynamical model rather than as an unrelated collection of one-time Stinespring representations. Second, the regularity of the map $t\mapsto U_t$ is essential. Even if each channel $\Phi_t$ admits a finite-dimensional dilation individually, it does not follow that the dilating unitaries can be selected in a way that is continuous, differentiable, or analytic in time. Thus the central issue is not existence in isolation, but existence together with temporal coherence.\\

This distinction explains the organization of the remainder of the chapter. In Section 4.2, we consider exact finite-dimensional dilations for analytic dynamical curves and show that analyticity of the reduced dynamics leads to analyticity of the corresponding dilation for positive times. In Section 4.3, we then examine a fundamental limitation of this construction, namely the singular behavior that may occur at the initial time $t=0$. Finally, in Section 4.4, we turn to approximate dilation results, where exact realization is replaced by arbitrarily accurate approximation and the regularity assumptions on the original curve are weakened from analyticity to Lipschitz continuity. Together, these results clarify both the scope and the limitations of the dilation-based approach in finite-dimensional quantum dynamics.

\section{Analytic Curves Admit Stinespring Dilations}
\begin{figure}
    \centering
    \includegraphics[width=0.5\linewidth]{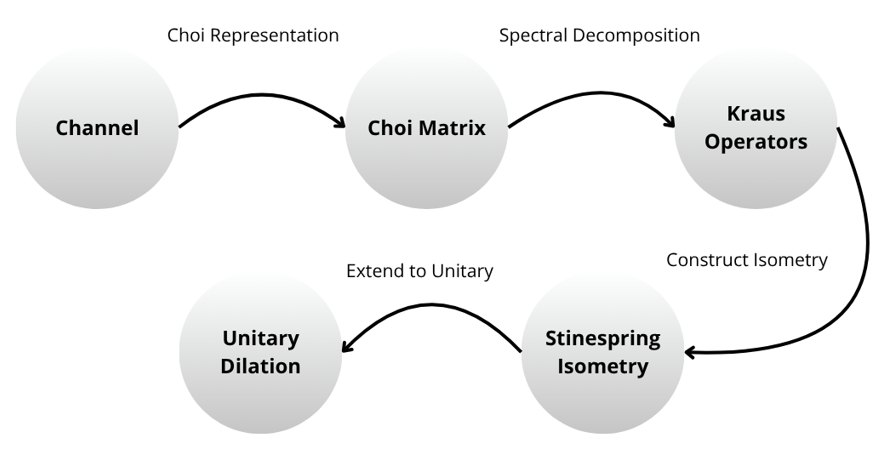}
    \caption{Proof Chain for Theorem 4.2}
    \label{fig:proof_chain}
\end{figure}
We now turn to the exact realization result for dynamical curves. The aim of this section is to show that, in finite dimensions, analyticity of the reduced channel-valued evolution is strong enough to produce a corresponding Stinespring dilation whose dilating unitary curve is analytic for all positive times. The result is based on the constructive method of Dive, Mintert, and Burgarth \cite{dive_mintert_burgarth_2015}, together with the finite-dimensional regularity-preserving passage from Stinespring isometries to unitary dilations emphasized in later work by vom Ende \cite{vom_ende_2024_stinespring}. The proof is important for the present thesis because it makes the dilation problem explicit at the level of an algorithm: starting from a curve of channels, one passes through the Choi representation, extracts time-dependent Kraus operators, and then builds a unitary dilation on an enlarged finite-dimensional space.

\begin{theorem}
    Let $n\in\mathbb{N}$, $\mathcal{I}=[0,t_f)$ with $t_f\in(0,\infty]$, and
    \[
    \Phi:\mathcal{I}\to\mathrm{CPTP}(n),\hspace{4mm}t\mapsto \Phi_t,
    \]
    be a dynamical curve which is analytic in $t$. Then there exists a finite-dimensional ancillary Hilbert space $E$ with
    \[
    \dim E\leq n^2,
    \]
    a unit vector $\omega\in E$, and a family of unitary operators
    \[
    U_t:\mathbb{C}^n\otimes E\to\mathbb{C}^n\otimes E,\hspace{4mm} t\in\mathcal{I},
    \]
    such that 
    \[
    \Phi_t(\rho)=\mathrm{Tr}_E\left(U_t(\rho\otimes\ket{\omega}\bra{\omega})U_t^\dagger\right)
    \]
    for all $\rho\in\mathcal{B}(\mathbb{C}^n)$, and the map $t\mapsto U_t$ is analytic on $(0,t_f)$.\\
    In particular, every finite-dimensional quantum dynamical semigroup admits such a dilation on positive times.
\end{theorem}
\begin{proof}
    The only nontrivial step beyond the Choi-Kraus construction is the extension of the resulting analytic Stinespring isometry to an analytic family of dilating unitaries. The proof proceeds by the chain in fig.~\ref{fig:proof_chain}, where each step is finite-dimensional and explicit.\\
    First, since $\Phi_t\in\mathrm{CPTP}(n)$ for each $t$, its Choi matrix
    \[
    J(\Phi_t)\in\mathcal{B}(\mathbb{C}^n\otimes\mathbb{C}^n)
    \]
    is well defined. By the discussion in Chapter 3, the assignment $\Phi_t\mapsto J(\Phi_t)$ is linear, so analyticity of $t\mapsto \Phi_t$ implies analyticity of $t\mapsto J(\Phi_t)$. Moreover, each $J(\Phi_t)$ is positive semidefinite and Hermitian.\\

    Because $J(\Phi_t)$ is a finite-dimensional Hermitian matrix-valued analytic function of $t$, the analytic perturbation theory of Hermitian matrices applies. In particular, after choosing a labeling of the spectral data appropriately, one may select eigenvalues and eigenvectors of $J(\Phi_t)$ analytically in $t$ on the real interval under consideration. Thus, for $t>0$, one may write
    \[
    J(\Phi_t)=\sum_{j=1}^{R_t}\lambda_j(t)\ket{v_j(t)}\bra{v_j(t)},
    \]
    where the eigenvalue functions $\lambda_j(t)$ and eigenvector functions $\ket{v_j(t)}$ are analytic, and where $R_t\leq n^2$ is the Kraus rank. This is the key spectral input in the construction.\\
    From this spectral decomposition one constructs Kraus operators exactly as in the Choi-to-Kraus conversion done in Chapter 3. Define
    \[
    \ket{\psi_j(t)}:=\sqrt{\lambda_j(t)}\ket{v_j(t)},
    \]
    and let
    \[
    K_j(t):=\mathrm{vec}^{-1}(\ket{\psi_j(t)}).
    \]
    Then
    \[
    \Phi_t(\rho)=\sum_{j=1}^{R_t}K_j(t)\rho K_j(t)^\dagger.
    \]
    Since the Choi-to-Kraus passage is constructive, this produces a family of Kraus operators directly from the spectral data of the Choi matrix. The only delicate point is the square root of the eigenvalues. For positive times, this still yields Kraus operators that are analytic in $t$, while the possible breakdown at $t=0$ is precisely the issue deferred to the next section.\\

    Next, fix an ancillary space $E\cong\mathbb{C}^{n^2}$ with orthonormal basis $e_1,\dots,e_{n^2}$, padding the Kraus family by zero operators if necessary so that it has exactly $n^2$ terms. Define
    \[
    V_t:\mathbb{C}^n\to\mathbb{C}^n\otimes E,\hspace{6mm}V_tx:=\sum_{j=1}^{n^2}K_j(t)x\otimes e_j.
    \]
    Because the Kraus operators satisfy the completeness relation
    \[
    \sum_{j=1}^{n^2}K_j(t)^\dagger K_j(t)=I,
    \]
    the map $V_t$ is an isometry for each $t$. Moreover, since the entries of the $K_j(t)$ are analytic for $t>0$, the map $t\mapsto V_t$ is analytic on $(0,t_f)$.\\

    To complete the construction, it remains to pass from the analytic family of isometries $V_t$ to an analytic family of unitaries on the enlarged space. We isolate this step as a lemma.

    \begin{lemma}
        Let $\mathcal{I}\subseteq\mathbb{R}$ be an interval, $\mathcal{H}$ and $\mathcal{K}$ be finite-dimensional Hilbert spaces, and
        \[
        V_t:\mathcal{H}\to\mathcal{H}\otimes \mathcal{K},\hspace{4mm} t\in\mathcal{I},
        \]
        be an analytic family of isometries. Then there exist a fixed unit vector $\omega\in \mathcal{K}$ and an analytic family of unitary operators
        \[
        U_t:\mathcal{H}\otimes \mathcal{K}\to\mathcal{H}\otimes \mathcal{K},\hspace{4mm}t\in\mathcal{I},
        \]
        such that
        \[
        V_tx=U_t(x\otimes\omega)
        \]
        for all $x\in\mathcal{H}$ and all $t\in\mathcal{I}$.
    \end{lemma}
    \begin{proof}
        Fix a unit vector $\omega\in \mathcal{K}$, and identify $\mathcal{H}$ with the subspace $\mathcal{H}\otimes\omega\subseteq\mathcal{H}\otimes \mathcal{K}$. For each $t$, the map $V_t$ is an isometric embedding of this fixed subspace into $\mathcal{H}\otimes\mathcal{K}$. Since $t\mapsto V_t$ is analytic, the corresponding family of range projections
        \[
        P_t:=V_tV_t^\dagger
        \]
        is analytic as well. In finite dimensions, an analytic family of isometries can be completed to an analytic family of unitaries. Equivalently, one may choose the complementary columns analytically so as to obtain a full unitary matrix depending analytically on $t$. Applying this completion to $V_t$ yields an analytic unitary family $U_t$ on $\mathcal{H}\otimes\mathcal{K}$ satisfying
        \[
        U_t(x\otimes\omega)=V_tx
        \]
        for all $x\in\mathcal{H}$. This regularity-preserving passage from Stinespring isometries to unitary dilations is standard in the finite-dimensional setting and is consistent with the equivalence results used by vom Ende \cite{vom_ende_2024_stinespring} for regular Stinespring curves.
    \end{proof}

    Returning to the proof of the theorem, apply the lemma to the analytic isometry family $V_t$. This yields an analytic unitary family $U_t$ on $\mathbb{C}^n\otimes E$ and a fixed unit vector $\omega\in E$ such that
    \[
    V_tx=U_t(x\otimes\omega).
    \]
    Therefore, for every $\rho\in\mathcal{B}(\mathbb{C}^n)$,
    \[
    \Phi_t(\rho)=\sum_{j=1}^{n^2}K_j(t)\rho K_j(t)^\dagger=\mathrm{Tr}_E\left(U_t(\rho\otimes\ket{\omega}\bra{\omega})U_t^\dagger\right).
    \]
    Hence $(U_t)_{t\in(0,t_f)}$ is an analytic Stinespring dilation of $\Phi$, this completes the construction.
\end{proof}

The proof above is constructive in a concrete sense. Starting from the dynamical curve $\Phi_t$, one first computes the Choi matrix $J(\Phi_t)$, then diagonalizes it analytically, then converts the resulting spectral data into Kraus operators, and finally packages those operators into a Stinespring isometry, which is completed analytically to a unitary dilation. This explicit chain is one of the main reasons the analytic setting is especially tractable. The following subsection illustrates this construction in a simple one-qubit example. At the same time, the proof also identifies the source of the later difficulty, namely, the square root appearing in the passage from eigenvalues of the Choi matrix to Kraus operators is precisely where singular behavior can enter. For that reason, the exact dilation theorem for positive times is naturally followed by a separate analysis of what happens at $t=0$.

\subsection{A toy example: pure dephasing}

We now illustrate the exact dilation construction on a simple one-qubit channel family. The purpose of this example is not to introduce new theory, but to carry out explicitly the same constructive chain used in the proof of Theorem 4.2:
\[
\Phi_t \;\longrightarrow\; J(\Phi_t)\;\longrightarrow\; \{K_1(t),K_2(t)\}\;\longrightarrow\; V_t\;\longrightarrow\; U_t.
\]

Consider the generator
\[
\mathcal{L}(\rho)=\frac{\gamma}{2}\bigl(\sigma_z\rho\sigma_z-\rho\bigr),
\]
and the associated channel curve
\[
\Phi_t=e^{t\mathcal{L}}.
\]
For
\[
\rho=
\begin{pmatrix}
\rho_{00} & \rho_{01}\\
\rho_{10} & \rho_{11}
\end{pmatrix},
\]
the channel acts as
\[
\Phi_t(\rho)=
\begin{pmatrix}
\rho_{00} & e^{-\gamma t}\rho_{01}\\
e^{-\gamma t}\rho_{10} & \rho_{11}
\end{pmatrix}.
\]
Thus the diagonal entries remain unchanged, while the off-diagonal entries decay exponentially. This is the standard pure-dephasing channel.

\stepparagraph{Step 1: Form the Choi matrix.}
Using the computational basis and matrix units \(\ket{i}\bra{j}\), define
\[
J(\Phi_t)=\sum_{i,j=0}^1 \ket{i}\bra{j}\otimes \Phi_t(\ket{i}\bra{j}).
\]
In this example,
\[
J(\Phi_t)=
\begin{pmatrix}
1 & 0 & 0 & e^{-\gamma t}\\
0 & 0 & 0 & 0\\
0 & 0 & 0 & 0\\
e^{-\gamma t} & 0 & 0 & 1
\end{pmatrix}.
\]

\stepparagraph{Step 2: Diagonalize the Choi matrix.}
Set
\[
a(t):=e^{-\gamma t}.
\]
Then the nonzero eigenvalues of \(J(\Phi_t)\) are
\[
\lambda_1(t)=1+a(t),\qquad \lambda_2(t)=1-a(t),
\]
with corresponding eigenvectors
\[
|v_1(t)\rangle=\frac{1}{\sqrt{2}}
\begin{pmatrix}
1\\0\\0\\1
\end{pmatrix},
\qquad
|v_2(t)\rangle=\frac{1}{\sqrt{2}}
\begin{pmatrix}
1\\0\\0\\-1
\end{pmatrix}.
\]
Hence
\[
J(\Phi_t)=\lambda_1(t)|v_1(t)\rangle\langle v_1(t)|
+\lambda_2(t)|v_2(t)\rangle\langle v_2(t)|.
\]

\stepparagraph{Step 3: Build Kraus operators from the spectral data.}
Define
\[
|\psi_j(t)\rangle=\sqrt{\lambda_j(t)}\,|v_j(t)\rangle.
\]
Then
\[
|\psi_1(t)\rangle
=
\sqrt{\frac{1+a(t)}{2}}
\begin{pmatrix}
1\\0\\0\\1
\end{pmatrix},
\qquad
|\psi_2(t)\rangle
=
\sqrt{\frac{1-a(t)}{2}}
\begin{pmatrix}
1\\0\\0\\-1
\end{pmatrix}.
\]
Unvectorizing gives
\[
K_1(t)=\operatorname{vec}^{-1}(|\psi_1(t)\rangle)
=
\sqrt{\frac{1+a(t)}{2}}
\begin{pmatrix}
1&0\\
0&1
\end{pmatrix},
\]
\[
K_2(t)=\operatorname{vec}^{-1}(|\psi_2(t)\rangle)
=
\sqrt{\frac{1-a(t)}{2}}
\begin{pmatrix}
1&0\\
0&-1
\end{pmatrix}.
\]
Therefore
\[
\Phi_t(\rho)=K_1(t)\rho K_1(t)^\ast+K_2(t)\rho K_2(t)^\ast.
\]

\stepparagraph{Step 4: Package the Kraus operators into an isometry.}
Let
\[
E=\mathbb{C}^2
\]
with orthonormal basis \(\{\ket{e_1},\ket{e_2}\}\). Define
\[
V_t x=K_1(t)x\otimes \ket{e_1}+K_2(t)x\otimes \ket{e_2}.
\]
Equivalently,
\[
V_t=
\begin{pmatrix}
K_1(t)\\
K_2(t)
\end{pmatrix}.
\]
Since the Kraus operators satisfy the completeness relation, one has
\[
V_t^\ast V_t=I.
\]
Moreover,
\[
\Phi_t(\rho)=\operatorname{tr}_E\!\left(V_t\rho V_t^\ast\right).
\]

\stepparagraph{Step 5: Complete the isometry to a unitary.}
Define
\[
\alpha(t):=\sqrt{\frac{1+e^{-\gamma t}}{2}},
\qquad
\beta(t):=\sqrt{\frac{1-e^{-\gamma t}}{2}}.
\]
Then
\[
V_t=
\begin{pmatrix}
\alpha(t) & 0\\
0 & \alpha(t)\\
\beta(t) & 0\\
0 & -\beta(t)
\end{pmatrix}
=
\begin{pmatrix}
\vert & \vert\\
v_1(t) & v_2(t)\\
\vert & \vert
\end{pmatrix},
\]
where
\[
v_1(t)=
\begin{pmatrix}
\alpha(t)\\
0\\
\beta(t)\\
0
\end{pmatrix},
\qquad
v_2(t)=
\begin{pmatrix}
0\\
\alpha(t)\\
0\\
-\beta(t)
\end{pmatrix}.
\]
Choose
\[
w_1(t)=
\begin{pmatrix}
-\beta(t)\\
0\\
\alpha(t)\\
0
\end{pmatrix},
\qquad
w_2(t)=
\begin{pmatrix}
0\\
\beta(t)\\
0\\
\alpha(t)
\end{pmatrix},
\]
which form an orthonormal basis of \(\operatorname{ran}(V_t)^\perp\). Then define
\[
U_t=\big[\,v_1(t)\;v_2(t)\mid w_1(t)\;w_2(t)\,\big].
\]
In matrix form,
\[
U_t=
\begin{pmatrix}
\alpha(t) & 0 & -\beta(t) & 0\\
0 & \alpha(t) & 0 & \beta(t)\\
\beta(t) & 0 & \alpha(t) & 0\\
0 & -\beta(t) & 0 & \alpha(t)
\end{pmatrix}.
\]

\stepparagraph{Step 6: Recover the channel from the unitary.}
By construction,
\[
U_t(x\otimes \ket{e_1})=V_t x.
\]
Hence the channel admits the Stinespring realization
\[
\Phi_t(\rho)=\operatorname{tr}_E\!\left[U_t(\rho\otimes \ket{e_1}\bra{e_1})U_t^\ast\right].
\]

This example shows the same constructive pipeline as in the proof, now carried out explicitly for a simple channel family.

\section{The singularity at $t=0$}
The analytic dilation result of Section 4.2 shows that an analytic dynamical curve admits a finite-dimensional Stinespring dilation whose dilating unitary curve is analytic for all positive times. It is natural to ask whether this regularity can be extended to the initial time $t=0$. In general, the answer is no. The construction of the previous section identifies a distinguished obstruction. Although the Choi matrix $J(\Phi_t)$ depends analytically on $t$, the passage from its spectral decomposition to Kraus operators introduces square roots of the eigenvalues. It is precisely this square-root step that can destroy differentiability at the initial time and thereby produce singular behavior in the corresponding dilation.\\

More concretely, if
\[
J(\Phi_t)=\sum_j\lambda_j(t)\ket{v_j(t)}\bra{v_j(t)}
\]
is the analytic spectral decomposition used in Section 4.2, then the corresponding Kraus operators are built from vectors of the form 
\[
\ket{\psi_j(t)}=\sqrt{\lambda_j(t)}\ket{v_j(t)}.
\]
Even when each $\lambda_j(t)$ is analytic, the derivative of $\sqrt{\lambda_j(t)}$ need not remain bounded at $t=0$. In particular, if $\lambda_j(0)=0$ but $\dot{\lambda}_j(0)\neq0$, then $\sqrt{\lambda_j(t)}$ behaves like $\sqrt{t}$ near the origin, and its derivative behaves like $t^{-1/2}$. Thus the Kraus operators remain continuous, but their derivatives may fail to exist or may diverge at the initial time. Since the dilating unitary curve is constructed from these Kraus operators, this loss of differentiability is inherited by the Hamiltonian
\[
H(t)=i\dot{U_t}U_t^\dagger,
\]
which may therefore become singular (unbounded) at $t=0$.\\

The phenomenon is not an artifact of poor choice of representation. Rather, it reflects a genuine limitation of exact finite-dimensional dilation for dissipative reduced dynamics. In this setting studied by Dive, Mintert, and Burgarth \cite{dive_mintert_burgarth_2015}, the Hamiltonian associated with the dilation can fail to remain bounded at the initial time precisely because the reduced dynamics begin to leave the identity channel in a genuinely dissipative direction. Intuitively, a closed finite-dimensional system evolving under a bounded Hamiltonian starts out unitarily and therefore reversibly, whereas a dissipative reduced evolution already exhibits irreversible behavior at first order in time. The singularity of the dilation Hamiltonian is the mechanism by which the enlarged unitary model compensates for that mismatch.\\

This discussion may be summarized as follows.
\begin{proposition}
    Let $(\Phi_t)_{t\in[0,t_f)}$ be an analytic dynamical curve with $\Phi_0=\mathrm{id}$, and let $(U_t)_{t>0}$ be a finite-dimensional Stinespring dilation obtained by the constructive procedure of Section 4.2. Then the curve $t\mapsto U_t$ need not extend differentiability to $t=0$. In particular, if an eigenvalue $\lambda_j(t)$ of the Choi matrix satisfies
    \[
    \lambda_j(0)=0\hspace{4mm}\text{ and }\hspace{4mm}\dot{\lambda}_j(0)\neq0,
    \]
    then the corresponding Kraus operator derivative may diverge at $t=0$, and the Hamiltonian $H(t)=i\dot{U}_tU_t^\dagger$ may become unbounded there.
\end{proposition}

The point of this proposition is not that every analytic dynamical curve develops such a singularity, but that the exact dilation procedure does not, in general, preserve regularity at the initial time. The obstruction arises already at the level of Choi eigenvalues and therefore appears before any particular unitary completion is chosen. In this sense, the singular behavior at $t=0$ is a structural feature of the exact construction rather than a byproduct of a specific choice of dilation.\\

For quantum dynamical semigroups, this observation has an especially clear interpretation. If the generator contains a nontrivial dissipative part, then the reduced dynamics depart from the identity channel in a way that cannot be realized by a bounded system-environment Hamiltonian that is regular at the initial time. Thus, although Section 4.2 guarantees the existence of an exact analytic dilation for positive times, one should not expect that dilation to arise from a Hamiltonian that remains regular at $t=0$ in the genuinely dissipative case. This is the essential limitation of the exact analytic theory and the reason that later approximation results become important.\\

The role of the present section is therefore conceptual as much as technical. Section 4.2 showed that analytic reduced dynamics can be lifted to enlarged unitary dynamics for positive times. The present section shows that this lifting can become singular at the initial time, even in finite dimensions. Section 4.4 will explain how one can weaken the goal from exact realization to approximation and thereby recover a much more flexible dilation theory for regular dynamical curves.

\section{Lipschitz Curves Admit Approximate Stinespring Dilations}
Section 4.2 showed that analytic dynamical curves admit exact finite-dimensional Stinespring dilations for positive times, while Section 4.3 explained why such exact constructions may become singular at the initial time. It is therefore natural to weaken the problem. Rather than asking for an exact realization of a prescribed dynamical curve by a regular unitary evolution on a larger space, one may ask whether such a realization exists approximately. In the finite-dimensional setting, this turns out to be possible under substantially weaker regularity assumptions. In particular, a recent result of vom Ende \cite{vom_ende_2024_stinespring} shows that every Lipschitz-continuous dynamical curve can be approximated arbitrarily well by a finite-dimensional Stinespring curve with controlled regularity. This provides a robust counterpart to the exact analytic theory developed in the previous two sections and shows that the dilation-based picture extends far beyond the analytic case.\\

Before stating the theorem, we record the notation of approximate dilation appropriate to the present setting.

\begin{definition}[Approximate Stinespring dilation]
    Let $\mathcal{I}\subseteq\mathbb{R}$ be a time interval, $\mathcal{H}_S\cong\mathbb{C}^n$, and
    \[
    \Phi:\mathcal{I}\to\mathrm{CPTP}(n),\hspace{4mm}t\mapsto\Phi_t,
    \]
    be a dynamical curve. We say that $\Phi$ admits an \emph{approximate Stinespring dilation} if for every $\epsilon>0$ there exists
    \begin{itemize}
        \item a finite-dimensional ancillary Hilbert space $E_\epsilon$,
        \item a fixed ancilla state $\omega_\epsilon\in\mathcal{D}(E_\epsilon)$, and
        \item a dynamical curve $\Phi^\epsilon:\mathcal{I}\to\mathrm{CPTP}(n)$ of the form
        \[
        \Phi_t^\epsilon(\rho)=\mathrm{Tr}_{E_\epsilon}\left(U_t^\epsilon(\rho\otimes\omega_\epsilon)(U_t^\epsilon)^\dagger\right),\hspace{4mm}t\in\mathcal{I},
        \]
        where $t\mapsto U_t^\epsilon$ is a locally absolutely continuous unitary curve on $\mathcal{H}_S\otimes E_\epsilon$,
    \end{itemize}
    such that
    \[
    ||\Phi-\Phi^\epsilon||_\text{sup}:=\sup_{t\in\mathcal{I}}||\Phi_t-\Phi_t^\epsilon||<\epsilon.
    \]
    In other words, $\Phi$ admits an approximate Stinespring dilation if it can be uniformly approximated by reduced dynamics arising from a regular unitary evolution on a larger finite-dimensional Hilbert space.
\end{definition}

The main approximation result may now be stated.

\begin{theorem}
    Let $n\in\mathbb{N}$, $t_f\in(0,\infty]$, and
    \[
    \Phi:[0,t_f)\to\mathrm{CPTP}(n)
    \]
    be a dynamical curve which is Lipschitz continuous with respect to the diamond norm; that is, suppose there exists $K_\Phi>0$ such that
    \[
    ||\Phi_{t_2}-\Phi_{t_1}||_\diamond\leq K_\Phi|t_2-t_1|\hspace{4mm} \text{ for all } t_1,t_2\in[0,t_f).
    \]
    Then for every $\epsilon>0$, the curve $\Phi$ admits an approximate Stinespring dilation $\Phi^\epsilon$ satisfying 
    \[
    ||\Phi-\Phi^\epsilon||_\text{sup}<\epsilon,
    \]
    with dilation space of dimension at most $4n^2$. Moreover, the ancilla state may be chosen pure, and if $t_f<\infty$, the unitary curve may in fact be chosen analytic. If, in addition, $\Phi_0=\mathrm{id}$, then one may choose the unitary curve so that $U_0=I$.
\end{theorem}
A full proof of this theorem is not reproduced here. Instead, we summarize the main idea of the construction, since it follows that the result is not merely existential but constructive in spirit.\\

The starting point is a discretization of the original Lipschitz curve $\Phi$. One chooses a sufficiently fine time mesh
\[
0,\delta,2\delta,\dots
\]
and samples the channels $\Phi(j\delta)$ along that mesh. At each sample point, one selects a finite-dimensional Stinespring isometry for the corresponding channel. Since each sampled channel is a single completely positive trace-preserving map on $\mathbb{C}^n$, such an isometry exists by the finite-dimensional Stinespring theorem. The central task is then to connect these sampled dilations in a coherent way so as to obtain a regular time-dependent unitary model on the larger Hilbert space. The Lipschitz assumption ensures that nearby sample points correspond to nearby channels, and this control allows one to interpolate between the associated dilations without losing uniform accuracy. In this way, the original curve is approximated by a new curve arising from a genuine Stinespring dilation with regular time dependence.\\

More concretely, the proof proceeds by a "connect-the-dots" strategy. After discretizing the interval, one chooses Stinespring isometries $V_j$ for the sampled channels $\Phi(j\delta)$. These isometries are then joined by carefully constructed paths on the enlarged space, producing a regular family $V_t^\epsilon$ of Stinespring isometries that agrees with the chosen samples at the mesh points and remains close to the original curve between them. A regularity-preserving completion step then promotes this isometry-valued curve to a unitary-valued curve $U_t^\epsilon$. Taking the partial trace over the finite ancilla yields an approximating dynamical curve $\Phi_t^\epsilon$ whose distance from $\Phi_t$ is uniformly controlled by the mesh size and the Lipschitz constant of the original dynamics. The theorem's dimension bound reflects the ancillary degrees of freedom needed to carry out this interpolation and completion procedure uniformly in time.\\

From the perspective of the present thesis, the significance of Theorem 4.5 is twofold. First, it shows that the finite-dimensional dilation framework remains meaningful even when the strong analyticity assumptions of Section 4.2 are dropped. Lipschitz continuity is a much weaker requirement, yet it still suffices to recover enlarged unitary models up to arbitrary accuracy. Second, the proof is constructive in a useful sense. It does not establish approximation by a purely abstract compactness argument, but rather by an explicit approximation-and-interpolation procedure. For this reason, the theorem supports the broader viewpoint that finite-dimensional dilation models may serve as tractable surrogate descriptions of reduced dynamics, even when exact regular dilations are unavailable.\\

This approximation result therefore complements the earlier exact theory. Section 4.2 showed that analytic dynamics admit exact finite-dimensional dilations for positive times, while Section 4.3 showed that the exact construction may develop singular behavior at the initial time. Theorem 4.5 shows that if one is willing to replace exact realization by arbitrarily accurate approximation, then a much broader class of dynamical curves can still be lifted to regular finite-dimensional unitary evolutions on enlarged Hilbert spaces. In this sense, approximation restores much of the flexibility lost in the exact setting and completes the picture developed in this chapter.

\chapter{Conclusion}

The central aim of this thesis has been to study quantum dynamics through a change of perspective. In the usual treatment of open quantum systems, one works directly with reduced dynamics, where the subsystem of interest is described effectively while the environment is suppressed. This remains one of the most useful and natural ways to model open-system behavior. The main point of the present thesis has been that this is not the only productive viewpoint. An alternative is to begin with the reduced evolution and ask whether it can be realized as part of a larger unitary dynamics on an extended Hilbert space. In this way, the problem shifts from describing dissipation and noise only at the subsystem level to understanding how such behavior may arise from a system-environment model in the first place.

Within the finite-dimensional setting adopted throughout this work, this question can be formulated cleanly in the language of quantum channels and dynamical curves. Here the parameter $n$ denotes the dimension of the system Hilbert space $\mathcal{H}_S \cong \mathbb{C}^n$. In particular, the framework includes the standard finite-dimensional models used in quantum information theory: a single qubit corresponds to $n=2$, an $m$-qubit register corresponds to $n=2^m$, and more generally a register of $m$ qudits of local dimension $d$ corresponds to $n=d^m$. In this sense, the finite-dimensional setting of the thesis is broad enough to include the usual small-register models of quantum computing and quantum noise, while remaining far simpler than the infinite-dimensional environments that often appear in microscopic open-system models.

This perspective organizes the main developments of the thesis. After introducing quantum channels and their Choi, Kraus, and Stinespring representations, time-dependent evolutions were reformulated as channel-valued dynamical curves. This made it possible to pose the dilation problem in a genuinely dynamical form: given a curve of quantum channels, when does there exist a correspondingly regular unitary dilation on a larger finite-dimensional Hilbert space? For analytic curves, the answer is positive on positive times: exact finite-dimensional Stinespring dilations exist, so a broad class of reduced dynamics can indeed be lifted to enlarged unitary models. At the same time, the singular behavior that may occur at $t=0$ shows that such reconstructions are not completely regular in general, and that this obstruction is tied to genuinely dissipative behavior rather than to a defect of the method.

The approximation results discussed later in the thesis extend this picture in an important way. Even when exact regular dilations are not available, Lipschitz-continuous dynamical curves may still be approximated arbitrarily well by finite-dimensional Stinespring curves. Thus the dilation-based viewpoint is not restricted to especially rigid or idealized dynamics. Instead, it remains meaningful for a much wider class of open-system evolutions, while still preserving a mathematically controlled relationship between the reduced description and an enlarged unitary one.

These results are primarily structural and mathematical, but they also fit naturally into a broader context. In quantum information and quantum computing, quantum channels are the standard language for noise, information transfer, and reduced subsystem evolution, with familiar examples including amplitude damping, phase damping, and depolarizing channels. From that point of view, the present thesis concerns time-dependent families of exactly the kinds of maps that already appear throughout the theory of quantum noise and error modeling. The finite-dimensional assumption is also natural from this perspective, since one often studies qubits, few-qubit registers, or other finite-level systems rather than infinite-dimensional baths.

A second point of context comes from quantum simulation. One motivation for dilation-based methods is that they offer a way to represent non-unitary reduced dynamics by unitary evolution on a larger but still finite-dimensional system. In the literature on simulation of dissipative dynamics, this idea appears explicitly as the problem of reproducing open-system channel evolution continuously in time by coupling the system to a finite ancilla and allowing time dependence in the larger Hamiltonian. From this point of view, the results discussed in this thesis help clarify when such enlarged finite-dimensional models exist exactly, when singular behavior can occur, and when approximation is the more flexible alternative.

A third point of context comes from control and noise modeling. Since reduced dynamics are often used to describe decoherence and dissipation, it is natural to ask how those same processes appear when lifted to a larger system-environment model. In that setting, the choice of dilation is not merely a formal matter; it can affect how one formulates questions about control, decoupling, and surrogate microscopic descriptions of noisy evolution. The present thesis does not develop control protocols, but it does contribute to the mathematical framework in which such questions can be posed.

Taken together, these results support the conclusion that reduced dynamics and dilation theory should be viewed as complementary rather than competing descriptions. The reduced picture remains indispensable when the goal is to model the effective evolution of a subsystem as efficiently as possible. The dilation picture becomes valuable when one wishes to understand how that effective evolution may be embedded into a larger dynamical structure. In this sense, the thesis has argued that a change of perspective can lead to new ways of formulating problems in open quantum systems. The main conclusion is therefore not that reduced dynamics should be replaced, but that finite-dimensional dilation theory provides a mathematically meaningful bridge between channel-based descriptions of open-system evolution and enlarged unitary models of the same dynamics.

\appendix
\renewcommand{\chaptername}{APPENDIX}
\renewcommand{\cftchappresnum}{Appendix\ }
\renewcommand{\cftchapaftersnum}{:}

\chapter{Mathematical and Quantum-Mechanical Preliminaries}
This appendix collects the mathematical and quantum-mechanical background used throughout the thesis. Its purpose is not to provide a comprehensive introduction to these subjects, but rather to establish the main definitions, notation, and standard facts needed in the main chapters of the thesis. Since this thesis is concerned entirely with finite-dimensional systems, all material is presented in that setting.\\

The discussion is intentionally brief and focused. Only those concepts that are used directly in the main text are included, and most results are stated without proofs unless a short argument is especially helpful for clarity. Readers seeking a more detailed treatment of the underlying mathematics or of the standard formalism of quantum mechanics may consult, for example, Nielsen and Chuang \cite{nielsen_chuang_2010}, Watrous \cite{watrous_2018}, or other standard references.

\section{Finite-Dimensional Hilbert Spaces and Notation}

Throughout this thesis, all Hilbert spaces are assumed to be finite-dimensional complex Hilbert spaces. This restriction is sufficient for all of the results considered here and allows the mathematical framework to remain entirely within the setting of finite-dimensional linear algebra. Unless stated otherwise, a Hilbert space will be denoted by \(\mathcal{H}\), and the algebra of linear operators on \(\mathcal{H}\) by \(\mathcal{B}(\mathcal{H})\).

The inner product on \(\mathcal{H}\) will be written as \(\langle \cdot,\cdot\rangle\), and vectors will frequently be expressed using Dirac's bra-ket notation. Thus, a vector in \(\mathcal{H}\) may be written as \(\ket{\psi}\), while its adjoint is written as \(\bra{\psi}\). Given vectors \(\ket{\psi},\ket{\phi}\in \mathcal{H}\), the rank-one operator \(\ket{\psi}\bra{\phi}\in \mathcal{B}(\mathcal{H})\) is defined by
\[
(\ket{\psi}\bra{\phi})(\ket{\chi})=\langle \phi,\chi\rangle \ket{\psi}
\]
for all \(\ket{\chi}\in \mathcal{H}\).

If \(\dim \mathcal{H} = n\), then upon choosing an orthonormal basis \(\{e_1,\dots,e_n\}\), the space \(\mathcal{H}\) may be identified with \(\mathbb{C}^n\), and each operator in \(\mathcal{B}(\mathcal{H})\) may be represented by an \(n\times n\) complex matrix. Although many constructions may therefore be written in matrix form, we will generally use basis-independent notation unless a particular choice of basis is convenient.

For Hilbert spaces \(\mathcal{H}_1\) and \(\mathcal{H}_2\), the notation \(\mathcal{H}_1 \cong \mathcal{H}_2\) indicates that they are isomorphic as Hilbert spaces. The identity operator on \(\mathcal{H}\) will be denoted by \(I_\mathcal{H}\), or simply by \(I\) when the underlying space is clear from context. We write \(\dim(\mathcal{H})\) for the dimension of \(\mathcal{H}\).

The finite-dimensional setting adopted here ensures that all linear operators are bounded, every subspace is closed, and every operator admits a matrix representation with respect to an orthonormal basis. Standard background on finite-dimensional Hilbert spaces and operator notation may be found in, for example, \cite{nielsen_chuang_2010,watrous_2018}.

\section{Linear Operators on Hilbert Spaces}

Let \(\mathcal{H}\) be a finite-dimensional Hilbert space. We write \(\mathcal{B}(\mathcal{H})\) for the set of all linear operators from \(\mathcal{H}\) to itself. Since \(\mathcal{H}\) is finite-dimensional, every linear operator on \(\mathcal{H}\) is automatically bounded.

For \(A \in \mathcal{B}(\mathcal{H})\), the \emph{adjoint} of \(A\) is the unique operator \(A^* \in \mathcal{B}(\mathcal{H})\) satisfying
\[
\langle A\psi,\phi\rangle = \langle \psi,A^*\phi\rangle
\]
for all \(\psi,\phi \in \mathcal{H}\). An operator \(A\) is called \emph{self-adjoint (Hermitian)} if \(A=A^*\), \emph{unitary} if \(A^*A=AA^*=I\), and \emph{positive} if
\[
\langle \psi,A\psi\rangle \geq 0
\]
for all \(\psi \in \mathcal{H}\). A self-adjoint idempotent operator \(P\in \mathcal{B}(\mathcal{H})\), that is, an operator satisfying \(P^2=P\) and \(P^*=P\), is called an \emph{orthogonal projection}. \\
Often, \(A^\dagger\) is used instead of \(A^*\) for the adjoint of an operator.

These classes of operators play a central role in quantum theory: self-adjoint operators describe observables and Hamiltonians, unitary operators describe closed-system evolution, and positive operators describe states and other physically meaningful quantities. Standard references include \cite{nielsen_chuang_2010,watrous_2018}.

\section{Trace, Operator Inner Products, and Basic Finite-Dimensional Facts}

If \(A \in \mathcal{B}(\mathcal{H})\) and \(\{e_1,\dots,e_n\}\) is an orthonormal basis for \(\mathcal{H}\), the \emph{trace} of \(A\) is defined by
\[
\mathrm{Tr}(A)=\sum_{i=1}^n \langle e_i,Ae_i\rangle.
\]
This quantity is independent of the choice of orthonormal basis. The trace is linear and satisfies the cyclicity relation
\[
\mathrm{Tr}(AB)=\mathrm{Tr}(BA)
\]
for all \(A,B\in \mathcal{B}(\mathcal{H})\).

The space \(\mathcal{B}(\mathcal{H})\) itself becomes a Hilbert space when equipped with the \emph{Hilbert--Schmidt inner product}
\[
\langle A,B\rangle_{\mathrm{HS}}=\mathrm{Tr}(A^*B).
\]
The corresponding norm is the Hilbert--Schmidt norm
\[
\|A\|_{\mathrm{HS}}=\sqrt{\mathrm{Tr}(A^*A)}.
\]
Since all spaces under consideration are finite-dimensional, many analytical complications do not arise: all norms are equivalent, every linear map is continuous, and every operator admits a matrix representation with respect to an orthonormal basis. Further background may be found in \cite{watrous_2018}.

\section{Tensor Products of Spaces and Operators}

If \(\mathcal{H}_1\) and \(\mathcal{H}_2\) are finite-dimensional Hilbert spaces, their tensor product is denoted by \(\mathcal{H}_1 \otimes \mathcal{H}_2\). This is again a finite-dimensional Hilbert space, and if \(\{e_i\}\) is an orthonormal basis for \(\mathcal{H}_1\) and \(\{f_j\}\) is an orthonormal basis for \(\mathcal{H}_2\), then \(\{e_i\otimes f_j\}_{i,j}\) is an orthonormal basis for \(\mathcal{H}_1\otimes \mathcal{H}_2\).

Given operators \(A \in \mathcal{B}(\mathcal{H}_1)\) and \(B \in \mathcal{B}(\mathcal{H}_2)\), the tensor product operator \(A\otimes B \in \mathcal{B}(\mathcal{H}_1\otimes \mathcal{H}_2)\) is defined on simple tensors by
\[
(A\otimes B)(\psi\otimes\phi)=A\psi\otimes B\phi,
\]
and extends linearly to all of \(\mathcal{H}_1\otimes \mathcal{H}_2\). The identity operator on a tensor product is written \(I_{\mathcal{H}_1}\otimes I_{\mathcal{H}_2}\), and operators of the form \(A\otimes I\) or \(I\otimes B\) are used to represent operators acting nontrivially on only one factor.

Tensor products provide the mathematical framework for composite quantum systems and for the enlarged system-environment spaces that appear throughout the thesis. Standard treatments may be found in \cite{nielsen_chuang_2010,watrous_2018}.

\section{Density Operators and Quantum States}

In the finite-dimensional formalism of quantum mechanics, the states of a system with Hilbert space \(\mathcal{H}\) are represented by \emph{density operators}. A density operator is an operator \(\rho \in \mathcal{B}(\mathcal{H})\) satisfying
\[
\rho \geq 0
\qquad \text{and} \qquad
\mathrm{Tr}(\rho)=1.
\]
The set of all density operators on \(\mathcal{H}\) will be denoted by \(\mathcal{D}(\mathcal{H})\).

A state is called \emph{pure} if it is of the form
\[
\rho = \ket{\psi}\bra{\psi}
\]
for some unit vector \(\psi \in \mathcal{H}\). More generally, a state is called \emph{mixed} if it can be written as a nontrivial convex combination of pure states. Thus, if \(\{\psi_i\}\subseteq \mathcal{H}\) are unit vectors and \(\{p_i\}\) is a probability distribution, then
\[
\rho=\sum_i p_i \ket{\psi_i}\bra{\psi_i}
\]
is a density operator.

This density-operator formalism is the natural setting for the treatment of both closed and open quantum systems. See \cite{nielsen_chuang_2010,watrous_2018} for further discussion.

\section{Closed-System Evolution}

For a closed quantum system with Hilbert space \(\mathcal{H}\), time evolution is described by a family of unitary operators \(U_t \in \mathcal{B}(\mathcal{H})\). In the density-operator formalism, the evolution of a state \(\rho\) is given by
\[
\rho \mapsto U_t \rho\, U_t^*.
\]
Such evolutions preserve positivity and trace, and therefore map density operators to density operators.

If the system is governed by a Hamiltonian \(H\), then the corresponding state evolution in the Schr\"odinger picture is described by the von Neumann equation
\[
\frac{d}{dt}\rho(t) = -i[H,\rho(t)],
\]
where
\[
[H,\rho]=H\rho-\rho H
\]
denotes the commutator. This provides the standard mathematical description of closed-system quantum dynamics. Background may be found in \cite{nielsen_chuang_2010,breuer_petruccione_2002}.

\section{Partial Trace and Reduced States}

Let \(\mathcal{H}_S\) and \(\mathcal{H}_E\) be finite-dimensional Hilbert spaces, interpreted respectively as the state spaces of a system and an environment. If \(\rho \in \mathcal{B}(\mathcal{H}_S\otimes \mathcal{H}_E)\), the \emph{partial trace over the environment} is the unique operator \(\mathrm{Tr}_E(\rho)\in \mathcal{B}(\mathcal{H}_S)\) characterized by
\[
\mathrm{Tr}\bigl(A\,\mathrm{Tr}_E(\rho)\bigr)=\mathrm{Tr}\bigl((A\otimes I_{\mathcal{H}_E})\rho\bigr)
\]
for all \(A\in \mathcal{B}(\mathcal{H}_S)\). Similarly, one defines the partial trace over the system, denoted \(\mathrm{Tr}_S\).

If \(\rho\) is a density operator on \(\mathcal{H}_S\otimes \mathcal{H}_E\), then
\[
\rho_S := \mathrm{Tr}_E(\rho)
\]
is again a density operator, called the \emph{reduced state} of the system. The partial trace is therefore the standard operation used to pass from the description of a composite system to that of a subsystem.

This construction is fundamental in the study of open quantum systems and reduced dynamics. Standard references include \cite{nielsen_chuang_2010,watrous_2018,breuer_petruccione_2002}.

\section{Spectral Decomposition and Positive Square Roots}

Many operator-theoretic constructions used in finite-dimensional quantum theory rely on the spectral theorem. In particular, every self-adjoint operator \(A \in \mathcal{B}(\mathcal{H})\) admits a spectral decomposition
\[
A=\sum_{j=1}^m \lambda_j P_j,
\]
where the \(\lambda_j\) are real numbers and the \(P_j\) are pairwise orthogonal projections satisfying \(\sum_j P_j = I\).

If \(A\) is positive, then all eigenvalues \(\lambda_j\) are nonnegative. In this case one may define
\[
A^{1/2}=\sum_{j=1}^m \lambda_j^{1/2} P_j.
\]
This operator is again positive and satisfies
\[
(A^{1/2})^2=A.
\]
Moreover, \(A^{1/2}\) is the unique positive square root of \(A\).

These standard facts will be used repeatedly in later operator constructions. For further details, see \cite{watrous_2018}.

\section{Regularity of Matrix-Valued Functions}

Later chapters consider time-dependent families of operators and maps. Since all spaces in this thesis are finite-dimensional, regularity questions may be treated entrywise in any matrix representation.

Let \(\mathcal{I}\subseteq \mathbb{R}\) be an interval and let \(A:\mathcal{I}\to \mathcal{B}(\mathcal{H})\) be a matrix-valued function. One says that \(A\) is \emph{continuous}, \emph{differentiable}, or \emph{analytic} if each matrix entry of \(A(t)\) has the corresponding property as a complex-valued function of \(t\). In finite dimensions, these notions are independent of the particular basis chosen.

Similarly, \(A\) is called \emph{Lipschitz continuous} on \(\mathcal{I}\) if there exists a constant \(K\geq 0\) such that
\[
\|A(t)-A(s)\| \leq K|t-s|
\]
for all \(s,t\in \mathcal{I}\), where \(\|\cdot\|\) may be any fixed norm on \(\mathcal{B}(\mathcal{H})\). Since all norms on a finite-dimensional space are equivalent, the precise choice of norm is immaterial for such qualitative regularity considerations.

These notions will be used when discussing regularity properties of time-dependent quantum evolutions and their dilations.

\clearpage
\UTSAVitaTOCEntry
\chapter*{VITA}
\thispagestyle{empty}
Caleb A. Mickelson is passionate about making complex ideas easier to understand. Although he did not begin taking academics seriously until later in life, he quickly developed a love for mathematics after first studying computer science. He earned his bachelor's degree in mathematics from the University of Texas at San Antonio in 2022 and later began teaching high school mathematics, where he came to value clear and thoughtful presentation of difficult concepts.

Driven by a strong desire to understand the universe at a fundamental level, Caleb has pursued physics largely through self-study. Before beginning this thesis, he had no formal background in quantum mechanics and taught himself the foundations of quantum systems along the way. That process deepened both his appreciation for mathematics and his interest in explaining difficult ideas in an accessible way. The past year has been deeply humbling. Having had the opportunity to work near the edge of research in quantum theory, he is honored to contribute, even in a small way, to the growing body of knowledge on quantum systems and looks forward to future understanding both within quantum theory and beyond.


\begin{thebibliography}{99}
\addcontentsline{toc}{chapter}{Bibliography}
\bibitem{breuer_petruccione_2002}
H.-P. Breuer and F. Petruccione,
\emph{The Theory of Open Quantum Systems},
Oxford University Press, Oxford, 2002.

\bibitem{nielsen_chuang_2010}
M. A. Nielsen and I. L. Chuang,
\emph{Quantum Computation and Quantum Information},
10th Anniversary ed.,
Cambridge University Press, Cambridge, 2010.

\bibitem{watrous_2018}
J. Watrous,
\emph{The Theory of Quantum Information},
Cambridge University Press, Cambridge, 2018.

\bibitem{kato_1995}
T. Kato,
\emph{Perturbation Theory for Linear Operators},
Classics in Mathematics,
Springer, Berlin, Heidelberg, 1995.

\bibitem{wiseman_milburn_2010}
H. M. Wiseman and G. J. Milburn,
\emph{Quantum Measurement and Control},
Cambridge University Press, Cambridge, 2010.

\bibitem{manzano_2020}
D. Manzano,
``A short introduction to the Lindblad master equation,''
\emph{AIP Advances} \textbf{10}, 025106 (2020).
doi:10.1063/1.5115323

\bibitem{dive_mintert_burgarth_2015}
B. Dive, F. Mintert, and D. Burgarth,
``Quantum simulations of dissipative dynamics: time-dependence instead of size,''
\emph{Physical Review A} \textbf{92}, 032111 (2015).
doi:10.1103/PhysRevA.92.032111

\bibitem{burgarth_facchi_hillier_2023}
D. Burgarth, P. Facchi, and R. Hillier,
``Control of Quantum Noise: On the Role of Dilations,''
\emph{Annales Henri Poincar\'e} \textbf{24}, 325--347 (2023).
doi:10.1007/s00023-022-01211-y

\bibitem{vom_ende_2023_church}
F. vom Ende,
``Quantum-Dynamical Semigroups and the Church of the Larger Hilbert Space,''
\emph{Open Systems \& Information Dynamics} \textbf{30}, 2350003 (2023).
doi:10.1142/S1230161223500038

\bibitem{vom_ende_2024_stinespring}
F. vom Ende,
``Finite-Dimensional Stinespring Curves Can Approximate Any Dynamics,''
\emph{Open Systems \& Information Dynamics} \textbf{31}, 2450004 (2024).
doi:10.1142/S1230161224500045

\bibitem{gorini_kossakowski_sudarshan_1976}
V. Gorini, A. Kossakowski, and E. C. G. Sudarshan,
``Completely Positive Dynamical Semigroups of N-Level Systems,''
\emph{Journal of Mathematical Physics} \textbf{17}(5), 821--825 (1976).
doi:10.1063/1.522979

\bibitem{lindblad_1976}
G. Lindblad,
``On the Generators of Quantum Dynamical Semigroups,''
\emph{Communications in Mathematical Physics} \textbf{48}(2), 119--130 (1976).
doi:10.1007/BF01608499

\bibitem{mit_ocw_open_quantum_systems}
P. Cappellaro,
``22.51 Course Notes, Chapter 8: Open Quantum Systems,''
\emph{MIT OpenCourseWare}, 22.51 Quantum Theory of Radiation Interactions, Fall 2012.
\url{https://ocw.mit.edu/courses/22-51-quantum-theory-of-radiation-interactions-fall-2012/resources/mit22_51f12_ch8/}

\bibitem{sanchez_palencia_2025}
L. Sanchez-Palencia,
``Lecture 6: Open quantum systems I: Kraus--Lindblad approach,''
lecture notes for \emph{Physics of Quantum Information}, Master QLMN, 17 October 2025.\\
\url{https://atom-tweezers-io.org/wp-content/uploads/2025/10/pqi2025-lecture6-open_quantum_systems.pdf}

\bibitem{stefanini_ziolkowska_budker_etal_2025}
M. Stefanini, A. A. Ziolkowska, D. Budker, U. Poschinger, F. Schmidt-Kaler,
A. Browaeys, A. Imamoglu, D. Chang, and J. Marino,
``Is Lindblad for me?''
\emph{arXiv preprint} arXiv:2506.22436 [quant-ph], 2025.

\end{thebibliography}
\end{document}